\documentclass[12pt]{amsart}

\usepackage{geometry}
\usepackage{amssymb}
\usepackage{amsmath}
\usepackage{graphicx}
\usepackage{enumerate}
\usepackage{multirow}
\usepackage{amsmath,color}
\usepackage{hyperref}
\usepackage[capitalize]{cleveref}
\usepackage{url}
\usepackage[colorinlistoftodos]{todonotes}
\usepackage{mathrsfs}
\usepackage{float}

\newcommand{\AG}{$\mathrm{AG}$}
\newcommand{\PG}{$\mathrm{PG}$}
\newcommand{\MA}{4-local arc}

\DeclareMathOperator{\F}{\mathbb{F}} 
 
\DeclareMathOperator{\MOLS}{\mathbb{L}} 
\DeclareMathOperator{\LS}{L} 
 
\DeclareMathOperator{\li}{\mathscr{L}}

\newtheorem{theorem}{Theorem}[section]

\newtheorem{lemma}[theorem]{Lemma}
\newtheorem{proposition}[theorem]{Proposition}
\newtheorem{definition}[theorem]{Definition}
\newtheorem{example}[theorem]{Example}
\numberwithin{equation}{section}
\theoremstyle{definition}
\newtheorem{remark}{Remark}[section]

\newcommand{\switchDisplayOrText}[2]{\mathchoice{#1}{#2}{#2}{#2}}
\NewDocumentCommand{\SmartSymbol}{o m m m}{
	\IfValueTF{#1}{
		\csname #1l\endcsname #2 #4 \csname #1r\endcsname #3
	}{
		\switchDisplayOrText{
			\left #2 #4 \right #3
		}{
			\if.#2\else #2\fi
			#4
			\if.#3\else #3\fi
		}
	}
}
\NewDocumentCommand{\card}{o m}{\SmartSymbol[#1]\lvert\rvert{#2}}

\begin{document}
	
\title[Singleton-optimal locally repairable codes]{The constructions of Singleton-optimal locally repairable codes with minimum distance 6 and locality 3}
\author{Yanzhen Xiong$^{a}$}

\author{Jianbing Lu$^{a,*}$}

\renewcommand{\thefootnote}{a}
\footnotetext{College of Science, National University of Defense Technology, Changsha 410073, China}

\renewcommand{\thefootnote}{*}
\footnotetext{Corresponding author \\Email addresses: xiongyanzhen@nudt.edu.cn, jianbinglu@nudt.edu.cn}

\begin{abstract}
In this paper, we present new  constructions of $q$-ary Singleton-optimal locally repairable codes (LRCs) with minimum distance $d=6$ and locality $r=3$, based on combinatorial structures from finite geometry. By exploiting the well-known correspondence between a complete set of mutually orthogonal Latin squares (MOLS) of order $q$ and the affine plane $\mathrm{AG}(2,q)$, We systematically construct families of disjoint 4-arcs in the projective plane $\mathrm{PG}(2,q)$, such that the union of any two distinct 4-arcs forms an 8-arc. These 4-arcs form what we call 4-local arcs, and their existence is equivalent to that of the desired codes. For any prime power $q\ge 7$, our construction yields codes  of length $n = 2q$, $2q-2$, or $2q-6$ depending on whether $q$ is even, $q\equiv 3 \pmod{4}$, or $q\equiv 1 \pmod{4}$, respectively. 
\newline
		
		\noindent\text{Keywords:} Locally repairable code; Singleton-type bound; Affine plane; Projective plane; Arc; Latin square
		
\end{abstract}

\maketitle

\section{Introduction}
Let $\mathbb{F}_{q}$ be the finite field with $q$ elements, where $q$ is a prime power.  An  $[n,k]_q$ \emph{linear code} $\mathcal{C}$ is a $k$-dimensional subspace of $\mathbb{F}_q^n$. The  parameter $n$ is called the  \emph{length} of  $\mathcal{C}$. The vectors in  $\mathcal{C}$ are called \emph{codewords}. Let $c\in\mathcal{C}$ be a codeword.   The \emph{support} of  $c=(c_1,c_2,\dots, c_n)$ is defined  as $\operatorname{supp}(c)=\{i\in[n]:c_i\neq0\}$, where $[n]$  denotes the set $\{1,2,\dots,n\}$. The \emph{weight} of $c$ is the cardinality of its support, i.e., $\operatorname{wt}(c)=\#\operatorname{supp}(c)$.  An $[n,k,d]_q$ linear code $\mathcal{C}$  is  an  $[n,k]_q$ linear code with \emph{minimum distance} $d$, defined as $d=\min_{0\neq c\in\mathcal{C}} wt(c)$. Usually, if the context is clear, the subscript $q$ is omitted  by convention in the sequel. The $i$-th symbol of $\mathcal{C}$    is said to have \emph{locality} $r$ if there exists a subset $R_i\subset [n]$ (called a \emph{repair group} for the symbol) containing $i$, with size $|R_i| \leq r+1$, such that $c_i$ can be expressed as a linear combination of the other symbols $\{c_j\}_{j\in R_i\setminus\{i\}}$.  A linear code $\mathcal{C}$ is called a \emph{locally repairable code} (LRC) with locality $r$ if each of its symbols has locality $r$. This concept was first introduced by Gopalan et al. \cite{gopalan2012}. We denote such a code as an $(n,k,d;r)$-LRC. An LRC is said to have disjoint local repair groups if the set of coordinates $[n]$ is exactly the disjoint union of some repair groups, i.e., $[n]=\bigcup_i R_i$, and $|R_i|=r+1$.

The Singleton bound \cite{Singleton1964} establishes the inequality $d \leq n - k + 1$ for any $[n,k,d]$ code. For an LRC with locality $r$, a refined Singleton-type bound was derived in \cite{gopalan2012}:
\begin{equation}\label{Singleton}
d\leq n-k-\left\lceil \frac{k}{r} \right\rceil+2  
\end{equation}
An $(n,k,d;r)$-LRC achieving equality in \eqref{Singleton}  is called  \emph{Singleton-optimal}.  There has been extensive work  on   constructing  Singleton-optimal LRCs. In 2014, Tamo and Barg \cite{tamo2014} first gave a construction with length $n$  up to  the alphabet size $q$. This was later improved to $q+1$ by Chen et al. using cyclic codes \cite{Chen2017} or constacyclic codes \cite{Chen2019}. In \cite{jin2019}, Jin et al. also  achieved length $q+1$ via the automorphism groups of rational function fields. For certain localities $r \in \{2,3,5,7,11,23\}$ with the additional condition $(r+1) \mid d$, constructions based on elliptic curves achieve lengths up to $q + 2\sqrt{q}$ \cite{li2018}. For  minimum distance $d=3$ or $4$, cyclic Singleton-optimal LRCs with unbounded length were proposed by Luo et al. \cite{Luo2018}. In  \cite{Jin2019}, Jin gave an explicit construction of $q$-ary  Singleton-optimal LRCs of length $\Omega_r(q^2)$ for $d=5,6$ and $r\geq4$. Xing and Yuan \cite{Xing2022} later generalized Jin’s results  to all $d \ge 7$ using hypergraph theory.

Certain Singleton-optimal LRCs   can be characterized by specific geometric objects in projective plane  $\mathrm{PG}(2,q)$,  (see, e.g., \cite{Chen2021,Fang2020,fang2024,Tao2025}). An \emph{$n$-arc} in $\mathrm{PG}(2,q)$ is a set of $n$ points such that  every  line intersects  the arc in at most $2$ points. It is known that $n\leq q+2$,  with equality holding if and only if $q$ is even. In \cite{fang2024}, the authors  provided sufficient and necessary conditions for the existence of  Singleton-optimal $(n,k,d;r)$-LRCs with $d=6,r=3$, and disjoint repair groups.
\begin{lemma}\label{LRCs}
 \cite[Theorem 6]{fang2024}   Suppose $4\mid n$. Then, there exists a $q$-ary Singleton-optimal LRC of length $n$, minimum distance $d=6$ and locality $r=3$ with disjoint repair groups if and only if there exist $n/4$ sets $\mathcal{S}_1, \mathcal{S}_2,\dots, \mathcal{S}_{n/4}$, each of which consists of $4$ points in $\mathrm{PG}(2,q)$, such that for any $1\leq i\neq j\leq n/4$,
\begin{enumerate}
	\item[(i)] $\mathcal{S}_i\cap \mathcal{S}_j=\emptyset$;
	\item[(ii)] no 3 points in $\mathcal{S}_i\cup \mathcal{S}_j$ lie on a common line.
\end{enumerate} 
\end{lemma}
By \cref{LRCs}, the existence of such a code is equivalent to the existence of a 4-local arc in $\mathrm{PG}(2,q)$; that is, a collection of $n/4$ disjoint $4$-arcs $\mathcal{S}_i$ such that the union of any two distinct $4$-arcs forms an $8$-arc.

A Latin square of order $n$ is an $n \times n$ matrix filled with $n$ distinct symbols, each occurring exactly once in each row and once in each column \cite{Dougherty20,DK15}.
Latin squares originate from Leonhard Euler's  \emph{``Thirty-Six Officers Problem''}, which is equivalent to the search for a pair of orthogonal Latin squares of order $6$. A set of Latin squares is called a set of \emph{mutually orthogonal Latin squares} (MOLS) if any two squares in the set are orthogonal; that is, when superimposed, every ordered pair of symbols appears exactly once. This concept is fundamental to many applications in design theory \cite{Bailey04}.
The connection between Latin squares and coding theory is fruitful \cite{DK91}. In particular, they serve as powerful tools for constructing low-density parity-check (LDPC) codes \cite{LM07,Vasic2002,ZHLAB10}, LRCs \cite{YER22}, error-distributing codes \cite{GP64}, and developing efficient secret-sharing schemes \cite{DK91}.

In this paper, we leverage this correspondence to provide new  explicit constructions of Singleton-optimal $(n,k,d=6,r=3)$-LRCs. Our approach starts from a complete set of MOLS of order $q$, from which we obtain the associated affine plane $\mathrm{AG}(2,q)$. We then systematically produce families of disjoint $4$-arcs in the projective plane $\mathrm{PG}(2,q)$.  The key combinatorial objects in our method are transversals of specially structured Latin squares—closely related to the addition table of $\mathbb{F}_q$—that directly provide point configurations in $\mathrm{AG}(2,q)$ satisfying the required "no‑three‑collinear" condition both inside each $4$-arc and between different 4-arcs. These collections of $4$-arcs precisely form $4$-local arcs (Definition \ref{4-local arc}) whose existence, by Lemma \ref{LRCs}, is equivalent to that of our target codes.  This approach offers a unified combinatorial framework for constructing these optimal codes, both for even and odd prime powers $q$. The following is our main result. 
\begin{theorem}\label{main}
    For any prime power $q\geq 4$, there exists a $q$-ary Singleton-optimal LRC with minimum distance $d=6$, locality $r=3$, and disjoint repair groups and with length
     \[
        n = \begin{cases}
            2q, & \text{if \ $q$ is even}; \\[4pt]
            2q-2, & \text{if\  $q \equiv 3 \pmod{4}$}; \\[4pt]
            2q-6, & \text{if \ $q \equiv 1 \pmod{4}$}.
        \end{cases}
    \]
\end{theorem}
The paper is organized as follows. \cref{pre} reviews preliminaries on LRCs with disjoint repair groups, affine and projective planes, and Latin squares. Section \ref{s3} discusses the well-known correspondence between a complete set of MOLS and the affine plane $\mathrm{AG}(2,q)$. Section \ref{s4} presents our constructions of 4-local arcs in $\mathrm{PG}(2,q)$ for both even and odd prime powers, thereby proving  Theorem \ref{main}. Finally, Section \ref{conclusion} concludes the paper with some remarks and open problems.

\section{Preliminaries}\label{pre}
In this section, we  present  some preliminary results about LRCs, affine  and projective planes,  and Latin squares.
\subsection{LRCs with disjoint repair groups}\label{disjoint}

We begin by recalling the structure of the parity-check matrix for LRCs with disjoint repair groups. The following lemma, based on \cite{Hao2020}, gives a canonical form.
\begin{lemma}\label{matrix}
 Suppose that $\mathcal{C}$ is a $q$-ary $(n,k,d;r)$-LRC with disjoint local repair groups such that $(r+1)\mid n$. Then $\mathcal{C}$ has an equivalent parity-check matrix $H$ of the following form:
\begin{equation}\label{parity-matrix}
\setlength{\arraycolsep}{1.5pt}
\renewcommand{\arraystretch}{0.9}
H = \left(
\begin{array}{cccc|cccc|c|cccc}  
1&1&\cdots&1&0&0&\cdots&0&\cdots&0&0&\cdots&0\\
0&0&\cdots&0&1&1&\cdots&1&\cdots&0&0&\cdots&0\\
\vdots&\vdots&\ddots&\vdots&\vdots&\vdots&\ddots&\vdots&\ddots&\vdots&\vdots&\ddots&\vdots\\
0&0&\cdots&0&0&0&\cdots&0&\cdots&1&1&\cdots&1\\
\hline
\mathbf{0}&\boldsymbol{v}_1^{(1)}&\cdots&\boldsymbol{v}_r^{(1)}&\mathbf{0}&\boldsymbol{v}_1^{(2)}&\cdots&\boldsymbol{v}_r^{(2)}&\cdots&\mathbf{0}&\boldsymbol{v}_1^{(\ell)}&\cdots&\boldsymbol{v}_r^{(\ell)}
\end{array}
\right),
\end{equation}
where the upper part of $H$ contains $\ell=n/(r+1)$ locality rows and the lower part of $H$ contains $u=n-k-\ell$ rows. The bold-type letters $\boldsymbol{v}_j^{(i)}$, $i\in[\ell]$ and $j\in[r]$ are column vectors in $\mathbb{F}_q^u$, and $\mathbf{0}$ represents the all-zero column vector in $\mathbb{F}_q^u$.
\end{lemma}

For the special case $d=6$ and $r=3$, we adopt the following geometric interpretation of the parity-check matrix \eqref{parity-matrix}. Denote the three nonzero column vectors in the $i$-th repair group by
\[
\boldsymbol{u}_i = \boldsymbol{v}_1^{(i)},\quad 
\boldsymbol{v}_i = \boldsymbol{v}_2^{(i)},\quad 
\boldsymbol{w}_i = \boldsymbol{v}_3^{(i)} \qquad (i=1,\dots,\ell).
\]
In the projective plane $\mathrm{PG}(2,q)$, we identify a non‑zero column vector $\boldsymbol{x}\in\mathbb{F}_q^{3}$ with the point $\langle \boldsymbol{x} \rangle$, and the span of two independent vectors $\langle \boldsymbol{x},\boldsymbol{y} \rangle$ with the line through the corresponding points. Accordingly, for each $i$ we define four lines:
\[
\begin{aligned}
L_{1,i} &= \langle \boldsymbol{u}_i, \boldsymbol{v}_i \rangle, &
L_{2,i} &= \langle \boldsymbol{v}_i, \boldsymbol{w}_i \rangle, \\
L_{3,i} &= \langle \boldsymbol{w}_i, \boldsymbol{u}_i \rangle, &
L_{4,i} &= \langle \boldsymbol{u}_i-\boldsymbol{v}_i,\; \boldsymbol{v}_i-\boldsymbol{w}_i \rangle .
\end{aligned}
\]
Set $\mathcal{B}_i = \{L_{1,i},L_{2,i},L_{3,i},L_{4,i}\}$. By \cite[Proposition 1]{fang2024}, for any $i\neq j$ the sets $\mathcal{B}_i$ and $\mathcal{B}_j$ are disjoint, and no three lines from $\mathcal{B}_i\cup\mathcal{B}_j$ are concurrent. Conversely, given a family $\{\mathcal{B}_i\}_{i=1}^{\ell}$ of quadruples of lines in $\mathrm{PG}(2,q)$ satisfying the same disjointness and non‑concurrency conditions, one can reverse the above correspondence to obtain a $q$-ary Singleton‑optimal $(4\ell,k,d;r)$-LRC with $d=6$, $r=3$ and disjoint repair groups \cite[Theorem 3]{fang2024}. This description is essentially the dual formulation of Lemma \ref{LRCs}.
To apply Lemma \ref{LRCs} directly, we introduce the following geometric object.
\begin{definition}[4-local arc]\label{4-local arc}
    Let $m \geq 2$ and let $\mathcal{S}_1, \mathcal{S}_2,\dots, \mathcal{S}_{m}$ be a collection   of $4$-arcs  in $\mathrm{PG}(2,q)$. If for every $i\neq j$ we have $\mathcal{S}_i\cap \mathcal{S}_j=\emptyset$ and $\mathcal{S}_i\cup \mathcal{S}_j$ is an $8$-arc, then  the union $\mathcal{S}_1\cup\mathcal{S}_2\cup\dots\cup \mathcal{S}_{m}$ is called a \emph{$4$-local arc}.
\end{definition}

\begin{remark}\label{rem:equivalence}
    By Lemma \ref{LRCs}, the existence of a $4$-local arc consisting of $m$ disjoint $4$-arcs is equivalent to the existence of a $q$-ary Singleton‑optimal $(n,k,d=6;r=3)$-LRC with disjoint repair groups, where the code length is $n = 4m$.
\end{remark}





\subsection{Affine planes and projective planes}\label{affine}
Now, we recall the fundamental definitions and constructions of the affine and projective planes over a finite field.
An \emph{affine plane} is an incidence structure of points and lines satisfying the following three axioms:
\begin{enumerate}
    \item   For any two distinct points, there is exactly one line through both;
    \item   Given any line $\ell$ and any point $P$ not on $\ell$, there is exactly one line through $P$ that does not meet $\ell$;
    \item  There exist four points such that no three are collinear.
\end{enumerate}
One constructs a  classical affine plane $\mathrm{AG}(2,q)$ by taking ordered pairs $(x,y)\in \F_q$ as points; and subsets of the form $y=mx+b$ or $x=a$ (for fixed $a,m,b\in \F_q$) as lines; that is, they are point sets of the form
\[ \{(x,mx+b):x\in \mathbb{F}_q\}\; \text{for} \; m,b\in \mathbb{F}_q;\quad \{(a,y):y\in \mathbb{F}_q\}\; \text{for} \; a\in \mathbb{F}_q. \]
The incidence is the natural one: a point is on a given line if and only if it satisfies the required linear equation. Two lines $\ell_1$ and $\ell_2$ in $\mathrm{AG}(2,q)$ are said to be \emph{parallel} if $\ell_1 = \ell_2$ or they have no point in common.  We can check that Parallelism is an equivalence relation on the lines of $\mathrm{AG}(2,q)$. For each $m\in \F_q\cup\{\infty\}$, lines of slope $m$ form a parallel class of lines  (Lines of slope $\infty$ are simply vertical lines of the form $x=$ constant).
The classical projective plane $\mathrm{PG}(2,q)$ can be obtained from $\mathrm{AG}(2,q)$ by the following standard extension:
\begin{enumerate}
    \item  Embed the affine points into the projective plane via the mapping $(x,y) \mapsto \langle (1,x,y) \rangle$;
    \item For each parallel class of lines with slope $m \in \mathbb{F}_q$, add a new \emph{point at infinity}, denoted by $\langle (0,1,m) \rangle$. For the vertical lines ($m=\infty$), add the point at infinity $\langle (0,0,1) \rangle$.
    \item Add a new line, called the \emph{line at infinity}, consisting of all the points at infinity added in step (2).
\end{enumerate}
Thus,  $\mathrm{PG}(2,q)$ can be constructed as follows: Take as points and lines the subspaces of $\mathbb{F}^{3}_q$ of dimension one and two, respectively. A point $P$ lies on a line $\ell$ if and only if $P\subset \ell$ as subspaces of $\mathbb{F}^{3}_q$.

\subsection{Latin square and MOLS}
Let $\mathcal{S}$ be a set of $n$ distinct symbols. A {\em Latin square} of order $n$, denoted by $\LS$ is an $n\times n$ matrix with entries from $\mathcal{S}$ such that for any distinct positions $(i,j)$ and $(i',j')$, it holds $\LS(i, j) \neq \LS(i', j')$ if either $i = i'$ or $j = j'$.

A {\em partial transversal} of $\LS$ is a set $\Gamma$ of  ordered pairs such that for any two distinct $(i,j)$ and $(i',j')$ in $\Gamma$, it holds that $i \neq i'$, $j \neq j'$ and $\LS(i,j) \neq \LS(i', j')$. The number of pairs in $\Gamma$ is called its \emph{length}.  Clearly, the length of a partial transversal cannot exceed $n$. A partial transversal of length $n$ is called a \emph{transversal}.

Let $\LS$ and $\LS'$ be two Latin squares of order $n$ with entries from $\mathcal{S}$. 
We say that $\LS$ and $\LS'$ are {\em orthogonal} if for every ordered pair $(s,s') \in \mathcal{S} \times \mathcal{S}$, there exists a unique position $(i,j)$ such that $\left(\LS(i, j), \LS'(i,j)\right) = (s,s')$.  
If they are orthogonal, we say that $\LS$ is an {\em orthogonal mate} of $\LS'$ or vice versa. A set of Latin squares of order $n$ is called a set of \emph{mutually orthogonal Latin squares} (MOLS) if every two distinct Latin squares in the set are orthogonal. When $n$ is a prime power, a set of MOLS is said to be \emph{complete} if it consists of exactly $n-1$ Latin squares. For basic properties and further background on Latin squares and MOLS, we refer to \cite{DK15}.

\section{From MOLS to affine plane}\label{s3}


Let $\omega$ be a primitive element of $\F_{q}$. For $k=1,2,\dots,q-1$, define a Latin square $\MOLS_k$  of order $q$ with rows and columns indexed by $\mathbb{F}_q$, whose entry in position $(i,j)$ is \[
\MOLS_k (i,j) = \omega^{\,k-1}\, i + j, \qquad i,j \in \mathbb{F}_q .
\]

\begin{lemma}\cite[Theorem~5.2.4]{DK15} \label{lem:falcon}
    The Latin squares $\MOLS_{k}$, $k=1,2,\dots,q-1$ form a complete set of MOLS of order $q$. 
\end{lemma}

\begin{remark}\label{rem:Cayley}
    Note that $\MOLS_1$ is  precisely the  Cayley table of the additive group  $(\mathbb{F}_q, +)$. Thanks to \cref{lem:falcon}, we know that the Cayley table can be extended to a complete set of MOLS.
\end{remark}

Let $\mathcal{L} = \{ \MOLS_1, \MOLS_2, \ldots, \MOLS_{q-1} \}$ be a complete set of MOLS of order $q$ defined in \cref{lem:falcon}. We extend this family by defining two additional $q \times q$ arrays:
\[
 \MOLS_\infty = 
\begin{bmatrix}
    0      & 0      & \cdots & 0      \\
    1      & 1      & \cdots & 1      \\
    \vdots & \vdots & \ddots & \vdots \\
    \omega^{q-2} & \omega^{q-2} & \cdots & \omega^{q-2}
\end{bmatrix},
\qquad
 \MOLS_0= 
\begin{bmatrix}
    0      & 1      & \cdots & \omega^{q-2} \\
    0      & 1      & \cdots & \omega^{q-2} \\
    \vdots & \vdots & \ddots & \vdots \\
    0      & 1      & \cdots & \omega^{q-2}
\end{bmatrix},
\]
where   $\omega$ is a primitive element of $\mathbb{F}_q$ and the rows and columns are indexed by $\mathbb{F}_q$. With these arrays we can describe the classical affine plane $\mathrm{AG}(2,q)$ in the following way (see \cite[Theorem 5.2.2]{DK15}):
\begin{enumerate}
    \item The point set is $\mathbb{F}_q \times \mathbb{F}_q$.
    \item For each $k \in \{0,1,\dots,q-1,\infty\}$ and each $\ell \in \mathbb{F}_q$, define a line
    \[
    \mathscr{L}(k,\ell) \;:=\; \bigl\{ (i,j) \in \mathbb{F}_q \times \mathbb{F}_q \mid \MOLS_k(i,j) = \ell \bigr\}.
    \]
\end{enumerate}
The collection of all such lines, together with the point set, forms the affine plane $\mathrm{AG}(2,q)$.

The following simple but useful proposition describes explicitly the unique line through two given distinct points in $\mathrm{AG}(2,q)$.

\begin{proposition}\label{prop:point-line}
    Let $(i, j)$ and $(i', j')$ be two distinct points of $\mathrm{AG}(2,q)$, and let $\ell$ be the unique line passing through both. Then:
    \begin{enumerate}
        \item If $i = i'$, then $\ell = \mathscr{L}(\infty, i)$, and the points on $\ell$ are $\{(i, k) \mid k \in \mathbb{F}_q\}$. 
        \item If $j = j'$, then $\ell = \mathscr{L}(0, j)$, and the points on $\ell$ are $\{(k, j) \mid k \in \mathbb{F}_q\}$.  
        \item If $i \neq i'$ and $j \neq j'$,  there exists a unique pair $(m, k) \in \{1,\dots,q-1\} \times \mathbb{F}_q$ such that 
              $\omega^{\,m-1} i + j = \omega^{\,m-1} i' + j' = k$, and $\ell = \mathscr{L}(m, k)$.
    \end{enumerate}
\end{proposition}
\section{Constructions of \MA s}\label{s4}

\subsection{When \texorpdfstring{$q$}{q} is a power of two}

Assume that $q = 2^r$ for some integer $r \geq 1$. The field $\mathbb{F}_{2^r}$ can be viewed as an $r$-dimensional vector space over $\mathbb{F}_2$. Fix a primitive element $\omega \in \mathbb{F}_{2^r}$; then $\{1, \omega, \omega^2, \dots, \omega^{r-1}\}$ is a basis of this vector space over $\mathbb{F}_2$.
Now arrange the elements of $\F_{2^r}$ as
\begin{align}\label{eq:order}
    0, 1, \omega, \omega + 1, \omega^2, \omega^2 + 1, \omega^2 + \omega, \omega^2+\omega+1, \ldots, \omega^{r-1} + \omega^{r-2} + \cdots + 1.
\end{align}
For every $r \geq 1$, let $\LS^{(r)}$ denote the Latin square of order $2^r$ induced by the Cayley table  of the additive group $(\mathbb{F}_{2^r}, +)$. We illustrate the recursive structure of these squares with the smallest cases:

\noindent
\textbf{Case $r=1$.} The Cayley table of $(\mathbb{F}_{2^1}, +)$ is
\[
\begin{array}{c|cc}
    + & 0 & 1 \\ \hline
    0 & 0 & 1 \\
    1 & 1 & 0
\end{array}
\qquad\text{giving}\qquad 
\LS^{(1)} = \begin{bmatrix}
    0 & 1 \\
    1 & 0
\end{bmatrix}.
\]

\noindent
\textbf{Case $r=2$.} Let $\omega$ be a primitive element of $\mathbb{F}_{2^2}$ satisfying $\omega^2 = \omega + 1$. 
The Cayley table of $(\mathbb{F}_{2^2}, +)$ is
\[
\begin{array}{c|cccc}
    + & 0 & 1 & \omega & \omega+1 \\ \hline
    0 & 0 & 1 & \omega & \omega+1 \\
    1 & 1 & 0 & \omega+1 & \omega \\
    \omega & \omega & \omega+1 & 0 & 1 \\
    \omega+1 & \omega+1 & \omega & 1 & 0
\end{array}
\qquad\text{giving}\qquad 
\LS^{(2)} = \begin{bmatrix}
    0 & 1 & \omega & \omega+1 \\
    1 & 0 & \omega+1 & \omega \\
    \omega & \omega+1 & 0 & 1 \\
    \omega+1 & \omega & 1 & 0
\end{bmatrix}.
\]
Inductively, ordering the elements of $\mathbb{F}_{2^r}$ as in \eqref{eq:order} yields a block decomposition of its Cayley table. 
Writing the elements in the order 
\[
0,\;1,\;\dots,\;\omega^{r-2}+\cdots+1 \quad\text{(first half)},\qquad 
\omega^{r-1},\;\dots,\;\omega^{r-1}+\cdots+1 \quad\text{(second half)},
\]
the Cayley table of $(\mathbb{F}_{2^r},+)$ takes the form
\[
\begin{array}{c|cc}
+ & \text{first half} & \text{second half} \\ \hline
\text{first half}  & \LS^{(r-1)} & \LS^{(r-1)} + \omega^{r-1} \\
\text{second half} & \LS^{(r-1)} + \omega^{r-1} &  \LS^{(r-1)}
\end{array}
\]
where $\LS^{(r-1)}$ is the Cayley table of $\mathbb{F}_{2^{r-1}}$ (with respect to the induced ordering) and the notation $+\omega^{r-1}$ means adding $\omega^{r-1}$ to every entry of the block. 
Consequently, the corresponding Latin square $\LS^{(r)}$ satisfies the recursive block construction

\begin{equation}\label{eq:recursive-L}
\LS^{(r)}= \begin{bmatrix}
 \LS^{(r-1)} & \LS^{(r-1)} + \omega^{r-1} \\[4pt]
    \LS^{(r-1)} + \omega^{r-1} & \LS^{(r-1)}
\end{bmatrix},
\qquad r\ge 2 .
\end{equation}

It follows from \cref{rem:Cayley} that $\LS^{(r)}$ can be extended to a complete MOLS, hence, it has an orthogonal mate when $r\geq 2$. By the following lemma, the existence of such a mate implies that $\LS^{(r)}$ has a transversal.
\begin{lemma}
\cite[Theorem 5.1.1]{DK15}  A Latin square of order $n$ has an orthogonal mate if and only if it has $n$ disjoint transversals. 
\end{lemma}

Now we are ready to construct a \MA \ in $\mathrm{PG}(2,2^{r})$ for every $r \geq 2$. Recall from Section \ref{affine} that points of the affine plane $\mathrm{AG}(2,2^{r})$ are embedded into the projective plane $\mathrm{PG}(2,2^{r})$ via $(x,y) \mapsto \langle (1,x,y) \rangle$.
In what follows we shall identify each affine point with its projective image and describe our configuration directly in $\mathrm{AG}(2,2^{r})$. Let $\Gamma$ be a transversal of the Latin square  $\LS^{(r-1)}$ defined in \eqref{eq:recursive-L}. 
For each $u = (i,j) \in \Gamma$, define a set of four points in $\mathrm{AG}(2,2^{r})$:
\begin{equation}\label{eq:2q}
    \Sigma_u \;:=\; 
    \bigl\{\, (i,j),\;(i+\omega^{r-1},\; j),\;(i,\; j+\omega^{r-1}),\;(i+\omega^{r-1},\; j+\omega^{r-1}) \,\bigr\}.
\end{equation}

\begin{theorem}\label{thm:evennumber}
    Let $q = 2^r$ with $r\geq 2$. Then, after embedding into $\mathrm{PG}(2,q)$,   $\bigcup_{u\in \Gamma} \Sigma_{u}$ is a \MA \ in \PG$(2,q)$ of size $2q$. 
\end{theorem}

\begin{proof}
    Since $\Gamma$ is a transversal of $\LS^{(r-1)}$, we have $\left| \Gamma \right| = 2^{r-1} = q/2$. 
    Then the size of $\bigcup_{u\in \Gamma} \Sigma_{u}$ is $4\left| \Gamma \right| = 2q$.  Recall that we identify affine points with their projective images. Hence, to prove that $\bigcup_{u\in \Gamma} \Sigma_{u}$ is a $4$-local arc in $\mathrm{PG}(2,q)$, it suffices to verify the defining properties of a $4$-local arc for the corresponding point sets in $\mathrm{AG}(2,q)$.
   Fix $u = (i,j)\in \Gamma$. The four points of $\Sigma_u$ are
    \[
    (i, j),\; (i+\omega^{r-1}, j),\; (i, j+\omega^{r-1}),\; (i+\omega^{r-1}, j+\omega^{r-1}).
    \]
 By \cref{prop:point-line}, the four points in $\Sigma_u$ determine the following six lines 
    \begin{align*}
        \li(\infty, i), \li(\infty, i+\omega^{r-1}), \li(0, j), \li(0, j+\omega^{r-1}), \li(1, i+j), \li(1, i+j+\omega^{r-1}). 
    \end{align*}
    Each of the six lines contains precisely two points of $\Sigma_u$. Hence, no three points of $\Sigma_u$ lie on a common line.

    Take another $v = (i', j') \in \Gamma$ with $v \neq u$.  We must show that no point of $\Sigma_v$ lies on any of the six lines  determined by  $\Sigma_u$.
    By the symmetry of the construction, it is enough to prove that the point  $v$ is not on  the three lines $\li(\infty, i), \li(0, j), \li(1, i+j)$.  Recall that $\Gamma$ is a transversal of $\LS^{(r-1)}$. Therefore, $i\neq i'$, $j\neq j'$ and $\LS^{(r)}(i,j) \neq \LS^{(r)}(i',j')$.  The first two inequalities immediately imply $(i', j') \notin  \li(\infty, i)$ and $(i', j') \notin\li(0, j)$. The third condition is equivalent to $i + j \neq i' + j'$, which proves $(i', j') \notin\li(1, i+j)$.  This completes the proof.
\end{proof}

We illustrate the general construction with the concrete case  $q=8$.

\begin{figure}[H]
    \centering
    \includegraphics[width=0.4\linewidth]{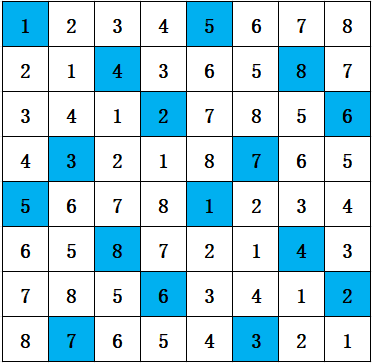}
    \caption{A \MA \ constructed from a Latin square of order $8$.}
    \label{fig:LS8}
\end{figure}

\begin{example}\label{ex:Zeitgeist}
        Consider the Latin square of order 8 shown in Figure~\ref{fig:LS8}. 
    Following the construction in Theorem \ref{thm:evennumber}, we obtain a 4-local arc consisting of four blocks
    \[
    \mathcal{S}_i = \{\,\text{blue cells containing symbols } i \text{ and } i+4\,\}, \qquad i=1,2,3,4.
    \]
     For instance, the four points of $\mathcal{S}_1$ determine six lines. In the Latin square, these correspond to the first and fifth rows, the first and fifth columns, and the cells marked with symbols 1 and 5.  One can visually verify that no other blue cell lies on any of these six lines, confirming the ``no‑three‑collinear'' condition across different blocks.
     
     Let $\omega$  be a primitive element of $\F_8$ satisfying $\omega^3+\omega+1=0$. 
      Identifying the symbols $(1,2,\dots ,8)$ with the field elements
    \[
    (0,\;1,\;\omega,\;\omega+1,\;\omega^{2},\;\omega^{2}+1,\;\omega^{2}+\omega,\;\omega^{2}+\omega+1),
    \]
    the square in Figure \ref{fig:LS8} coincides with $\LS^{(3)}$ – the Cayley table of $(\mathbb{F}_{2^{3}},+)$.  Under this identification, the four blocks become the following sets of projective points (written in homogeneous coordinates):
\[\mathcal{S}_1=\{\langle(1,0,0)\rangle,\langle(1,0,\omega^2)\rangle,\langle(1,\omega^2,0)\rangle,\langle(1,\omega^2,\omega^2)\rangle\},\]
\[\mathcal{S}_2=\{\langle(1,1,\omega)\rangle,\langle(1,1,\omega^4)\rangle,\langle(1,\omega^6,\omega)\rangle,\langle(1,\omega^6,\omega^4)\rangle\},\]
\[\mathcal{S}_3=\{\langle(1,\omega,\omega^3)\rangle,\langle(1,\omega,\omega^5)\rangle,\langle(1,\omega^4,\omega^3)\rangle,\langle(1,\omega^4,\omega^5)\rangle\},\]
\[\mathcal{S}_4=\{\langle(1,\omega^3,1)\rangle,\langle(1,\omega^3,\omega^6)\rangle,\langle(1,\omega^5,1)\rangle,\langle(1,\omega^5,\omega^6)\rangle\}.\]
Applying the correspondence described in Section \ref{disjoint} to this  4-local arc (also see \cite[Algorithm 2]{fang2024}) yields the following parity‑check matrix of a Singleton‑optimal $(16,9,6;3)$-LRC over $\mathbb{F}_8$:
\begin{equation*}
\setlength{\arraycolsep}{1.5pt}
\renewcommand{\arraystretch}{0.9}
H = \left(
\begin{array}{cccc|cccc|cccc|cccc}  
1&~1~&~1~&~1~&~0~&~0~&~0~&~0~&~0~&~0~&~0~&~0~&~0~&~0~&~0~&0\\0&0&0&0&1&1&1&1&0&0&0&0&0&0&0&0\\0&0&0&0&0&0&0&0&1&1&1&1&0&0&0&0\\0&0&0&0&0&0&0&0&0&0&0&0&1&1&1&1\\
\hline 0&0&0&\omega^2&0&\omega&1&\omega^5&0&\omega^3&\omega&\omega^6&0&1&\omega^3&\omega^4\\0&0&1&1&0&0&1&1&0&0&1&1&0&0&1&1\\0&1&0&1&0&1&0&1&0&1&0&1&0&1&0&1
\end{array}
\right).
\end{equation*}
\end{example}

\subsection{When  \texorpdfstring{$q$}{q} is an odd prime power}\label{even}
We now turn to the case where $q$ is an odd prime power. We first recall a basic fact about cyclic Latin squares. 

A Latin square of order $n$ is called {\em cyclic} if each row is a cyclic shift of its previous row. Up to relabelling of symbols, a cyclic Latin square $\LS$ can be written as
\[
\begin{bmatrix}
1 & 2 & 3 & \cdots & n\\
n & 1 & 2 & \cdots & n-1\\
n-1 & n & 1 & \cdots & n-2\\
\vdots & \vdots & \vdots & \ddots & \vdots\\
2 & 3 & 4 & \cdots & 1
\end{bmatrix}.
\]

\begin{lemma}\label{lem:partialtransversal}
    Let   $\LS$ be a cyclic Latin square of order $n$, defined by $\LS(i,j) = j-i+1 \pmod{n}$ for $i,j \in \mathbb{Z}_n$.
    \begin{enumerate}
        \item If  $n$ is odd, then $\LS$  has a transversal.
        \item If $n$ is even, then $\LS$  has a partial transversal of length $n-1$.
    \end{enumerate}
\end{lemma}

\begin{proof}
    We explicitly construct the required set of cells.
    
    \textbf{Case 1: $n = 2k+1$ odd.}
    Define
    \[
    \Gamma = \{\, (i, 2i) \mid i \in \mathbb{Z}_n \,\}.
    \]
     Since $n$ is odd, the map $i \mapsto 2i \pmod{n}$ is a permutation of $\mathbb{Z}_n$. 
    Hence the $n$ pairs in $\Gamma$ have distinct first coordinates and distinct second coordinates.
    Moreover, for each $(i,2i)\in \Gamma$, it holds that $\LS(i,2i) = i+1$. Therefore, $\Gamma$ is a transversal of $\LS$. 
   
    \textbf{Case 2: $n = 2k$ even.}
    Define
    \[  \Gamma = \{ (i, 2i) : \ i = 0,1,\ldots, k-1 \} \cup \{ (i+1, 2i-2k+1) : \ i = k, \ldots, 2k-2 \} \]
    as a set of size $n-1$. We can check that for any two distinct $(i,j)$ and $(i',j')$ in $\Gamma$,  $i \neq i'$ and $j \neq j'$ hold. Moreover,  $\LS(i,j) = i+1$ holds for any $(i,j) \in \Gamma$. 
   Thus,  $\Gamma$ is a partial transversal of $\LS$ of length $n-1$. This completes the proof. 
\end{proof}

Now we are ready to describe our construction of $4$-local arcs for odd $q$. 

 Write $\F_q = \{ 0, 1, \omega, \omega^2, \ldots, \omega^{q-2} \}$
and $\F_q^\ast = \F_q\setminus\{0\}$, where $\omega$ is a primitive element of $\F_q$. 
Let $\{ \MOLS_m : \ m = 0,1,\ldots, q-1, \infty \}$, $\li(k,\ell)$ and \AG$(2,q)$ be as defined in \cref{s3}. 
We define a $(q-1)\times (q-1)$ matrix $M$ whose rows and columns are indexed by the elements of $\F_q^\ast = \{ 1, \omega, \omega^2, \ldots, \omega^{q-2} \}$. 
For any $i,j \in \F_q^\ast$, we set $M(i,j)=m \in \{ 1,\ldots, \frac{q-1}{2} \}$ if the point $(i,j)$ lies on either $\li(m,0)$ or $\li\left(m+\frac{q-1}{2}, 0\right)$. 
Therefore, 
\begin{align}\label{eq:pinnacle}
   M = \begin{pmatrix}
L & L \\
L & L
\end{pmatrix},
\end{align}
where  $L$ is precisely a cyclic Latin square of order $\frac{q-1}{2}$ given by
\[
L = 
\begin{bmatrix}
1 & 2 & 3 & \cdots & \frac{q-1}{2} \\
\frac{q-1}{2} & 1 & 2 & \cdots & \frac{q-1}{2}-1 \\
\frac{q-1}{2}-1 & \frac{q-1}{2} & 1 & \cdots & \frac{q-1}{2}-2 \\
\vdots & \vdots & \vdots & \ddots & \vdots \\
2 & 3 & 4 & \cdots & 1
\end{bmatrix}.
\]
For convenience we refer to the first $\frac{q-1}{2}$ rows and columns of $M$ as the ``first block'', corresponding to the elements $1,\omega,\dots,\omega^{(q-3)/2}$.
Thus $L$ itself is the submatrix of $M$ restricted to the first block.

Let $\Gamma$ be a partial transversal of the cyclic Latin square $L$. 
For each $u = (i,j) \in \Gamma$, define a set of four points in $\mathrm{AG}(2,q)$: 
    \begin{align}\label{eq:4pts}
        \Sigma_u \;:=\; \bigl\{\, (i, j),\; (-i, j),\; (i, -j),\; (-i, -j) \,\bigr\}.
    \end{align}
The union  $\bigcup_{u\in \Gamma} \Sigma_{u}$ is our  proposed $4$-local arc after embedding into $\mathrm{PG}(2,q)$.

\begin{theorem}\label{thm:thimble}
   Let $q$ be an odd prime power.  Then, after embedding into $\mathrm{PG}(2,q)$,    $\bigcup_{u\in \Gamma} \Sigma_{u}$ is a \MA \ in \PG$(2,q)$ of size 
        \begin{align}\label{eq:Aplomb}
        n=\left\{
        \begin{array}{ll}
            2q-2, & \text{if } q \equiv 3 \pmod{4}; \\[4pt]
        2q-6, & \text{if } q \equiv 1 \pmod{4}.
        \end{array}
        \right.
    \end{align}
\end{theorem}

\begin{proof}
    Write $r = \frac{q-1}{2}$. 
    By \cref{lem:partialtransversal}, the cyclic Latin square $L$ of order $r$ admits a transversal when $r$ is odd, and a partial transversal of length $r-1$ when $r$ is even. 
    Therefore, the size of $\bigcup_{u\in \Gamma} \Sigma_{u}$ is $4\left| \Gamma \right|$, which equals \eqref{eq:Aplomb}.

    It remains to prove it is a \MA. 
    Fix $u = (i,j) \in \Gamma$.   By \cref{prop:point-line}, the six lines determined by the four points of $\Sigma_u$ are
    \begin{align}\label{eq:jubilee}
       \mathscr{L}(0,j),\; \mathscr{L}(0,-j),\; 
        \mathscr{L}(\infty,i),\; \mathscr{L}(\infty,-i),\; 
        \mathscr{L}(m,0),\; \mathscr{L}\!\bigl(m+r,\,0\bigr),
    \end{align}
    where $m = M(i,j)$. 
    As a result, no three points of $\Sigma_u$ are collinear.
   Take another $v=(i',j')\in\Gamma$ with $v\neq u$.  We must show that no point of $\Sigma_v$ lies on any of the six lines \eqref{eq:jubilee} determined by $\Sigma_u$. By symmetry it suffices to check the point $v=(i',j')$ itself.
    Since $\Gamma$ is a partial transversal of $L$,  it follows that $i\neq i'$, $j\neq j'$ and $L(i,j) \neq L(i',j')$. Further notice that $i,i' \in \{ 1,\omega,\ldots,\omega^{\frac{q-1}{2}-1} \}$, thus, we have $i'\neq-i$. Analogously, we have $j'\neq -j$. 
    Above all, we can prove  that the point $v$ does not lie on any of the six lines from \eqref{eq:jubilee}.  This completes the proof. 
\end{proof}

The construction for odd $q$ follows the same blueprint as for even $q$; therefore we omit a detailed geometric explanation analogous to  \cref{ex:Zeitgeist}.
We present two concrete examples.
\begin{example}

For $q=7$, we have $n=2q-2=12$.  Following the construction in Section \ref{even}, we obtain a $4$-local arc consisting of three blocks: 
\[\mathcal{S}_1=\{\langle(1,1,1)\rangle,\langle(1,1,6)\rangle,\langle(1,6,1)\rangle,\langle(1,6,6)\rangle\},\]
\[\mathcal{S}_2=\{\langle(1,3,2)\rangle,\langle(1,3,5)\rangle,\langle(1,4,2)\rangle,\langle(1,4,5)\rangle\},\]
\[\mathcal{S}_3=\{\langle(1,2,3)\rangle,\langle(1,2,4)\rangle,\langle(1,5,3)\rangle,\langle(1,5,4)\rangle\}.\]
The corresponding parity-check matrix of a Singleton‑optimal $(12,6,6;3)$-LRC over $\mathbb{F}_7$ is 

\begin{equation*}
\setlength{\arraycolsep}{1.5pt}
\renewcommand{\arraystretch}{0.9}
H = \left(
\begin{array}{cccc|cccc|cccc}  
~1~&~1~&~1~&~1~&~0~&~0~&~0~&~0~&~0~&~0~&~0~&~0~\\0&0&0&0&1&1&1&1&0&0&0&0\\0&0&0&0&0&0&0&0&1&1&1&1\\
\hline0&6&6&0&0&5&5&0&0&2&2&0\\0&0&1&1&0&0&3&3&0&0&6&6\\0&1&0&1&0&1&0&1&0&4&0&4
\end{array}
\right).
\end{equation*}
\end{example}


\begin{example}
For $q=9$, we have $n=2q-6=12$. Let $\delta$  be a primitive element of $\F_9$ satisfying  $\delta^2+2\delta+2=0$. The construction produces the following three blocks:
\[\mathcal{S}_1=\{\langle(1,1,1)\rangle,\langle(1,1,2)\rangle,\langle(1,2,1)\rangle,\langle(1,2,2)\rangle\},\]
\[\mathcal{S}_2=\{\langle(1,\delta,\delta^2)\rangle,\langle(1,\delta,\delta^6)\rangle,\langle(1,\delta^5,\delta^2)\rangle,\langle(1,\delta^5,\delta^6)\rangle\},\]
\[\mathcal{S}_3=\{\langle(1,\delta^3,\delta)\rangle,\langle(1,\delta^3,\delta^5)\rangle,\langle(1,\delta^7,\delta)\rangle,\langle(1,\delta^7,\delta^5)\rangle\}.\]
The corresponding parity-check matrix of a Singleton‑optimal $(12,6,6;3)$-LRC over $\mathbb{F}_9$ is 

\begin{equation*}
\setlength{\arraycolsep}{1.5pt}
\renewcommand{\arraystretch}{0.9}
H = \left(
\begin{array}{cccc|cccc|cccc}  
~1~ & ~1~ & ~1~ & ~1~ & ~0~ & ~0~ & ~0~ & ~0~ & ~0~ & ~0~ & ~0~ & ~0~ \\
0 & 0 & 0 & 0 & 1 & 1 & 1 & 1 & 0 & 0 & 0 & 0 \\
0 & 0 & 0 & 0 & 0 & 0 & 0 & 0 & 1 & 1 & 1 & 1 \\
\hline
0 & 1 & 1 & 0 & 0 & \delta^6 & \delta^6 & 0 & 0 & 1 & 1 & 0 \\
0 & 0 & 2 & 2 & 0 & 0 & \delta & \delta & 0 & 0 & \delta & \delta \\
0 & 2 & 0 & 2 & 0 & 1 & 0 & 1 & 0 & \delta^3 & 0 & \delta^3
\end{array}
\right).
\end{equation*}
\end{example}

\begin{proof}[Proof of Theorem \ref{main}]
    Theorem \ref{main} asserts the existence of Singleton-optimal $(n,k,6;3)$-LRCs for every prime power $q \geq 4$, with the length $n$ given by three cases according to the parity of $q$ and its residue modulo $4$. For even $q = 2^r$ ($r \ge 2$), Theorem \ref{thm:evennumber} provides a $4$-local arc in $\mathrm{PG}(2,q)$ of size $2q$, which, by Remark \ref{rem:equivalence}, yields a code of length $n = 2q$. For odd $q$, we distinguish two subcases. If $q \equiv 3 \pmod{4}$, Theorem \ref{thm:thimble} gives a $4$-local arc of size $2q-2$; if $q \equiv 1 \pmod{4}$, it gives a $4$-local arc of size $2q-6$. In both subcases, Remark \ref{rem:equivalence} again translates the geometric object into a Singleton-optimal LRC with the claimed length. This completes the proof.
\end{proof}

\section{Conclusion and discussion}\label{conclusion}

In \cite[Theorem 5]{fang2024}, the following upper bound was established for the length of a $q$-ary Singleton-optimal LRC with $d=6$ and $r=3$ (assuming $4 \mid n$):
\[
n \leq 4 \left\lfloor \frac{7q + 3 + \sqrt{24q^3 + q^2 - 6q - 63}}{24} \right\rfloor = O(q^{1.5}).
\]
The constructions presented in Theorems \ref{thm:evennumber} and \ref{thm:thimble} achieve lengths that grow linearly in $q$, which is considerably below this asymptotic bound. 
Moreover, when $q=9$, our method yields $n=12$, whereas the upper bound evaluates to $16$ (see \cite[Example~1]{fang2024}). 
This substantial gap suggests that either the known upper bound could be improved, or new constructions approaching the bound may exist. Investigating this discrepancy is therefore an interesting direction for future work.

A natural combinatorial extension of our work is to replace the $4$-arc blocks by $k$-arc blocks for a fixed integer $k \ge 4$, while  requiring that the union of any two distinct blocks forms an $2k$-arc.  
We call such a configuration a \emph{$k$-local arc} in $\operatorname{PG}(2,q)$.

Several purely geometric and combinatorial questions arise immediately:
\begin{itemize}
    \item \textbf{Existence and construction:} For which pairs $(q,k)$ do $k$-local arcs exist? Can one give explicit constructions?
    \item \textbf{Maximal size:} What is the maximum possible number of blocks (or equivalently, the maximum total number of points) in a $k$-local arc? How does this maximum grow with $q$ and $k$?
     \item \textbf{Potential coding-theoretic implications:} The geometric study of $k$-local arcs could inspire the construction of new linear codes. Exploring this connection remains an open direction at the intersection of finite geometry and coding theory.
\end{itemize}
Studying $k$-local arcs therefore offers a natural next step in understanding the combinatorial landscape of arc-type configurations in projective planes, independent of their potential coding-theoretic applications.

Finally, LRCs with different parameters $d$ and $r$ correspond to various geometric objects.  For example, \cite[Theorem 10]{Tao2025} shows that the existence of a $q$-ary Singleton-optimal $(n=\ell(\delta+1), k, d=2\delta+2, r=2, \delta)$-LRC with disjoint local repair groups is equivalent to the existence of a set of points in $\mathrm{PG}(2,q)$ with certain properties on secant and tangent lines.  Exploring the connections between LRCs with various parameters and finite geometric objects is also worth further attention.

\section{Acknowledgements}
We are grateful to Peter Cameron and Rosemary Bailey for teaching us a useful fact on mutually orthogonal Latin squares.  The first author is supported by the National Natural Science Foundation of China (No. 12501502) and Innovation Research Foundation of College of Science at National University of Defense Technology (202501-YJRC-LXY-01).  The second author is supported by the National Natural Science Foundation of China (No. 12501467) and Innovation Research Foundation of College of Science at National University of Defense Technology (202501-YJRC-LXY-02).


\end{document}